\documentclass[screen,sigconf,nonacm]{acmart}

\newtheorem{thm}{Theorem}[section]
\newtheorem{prop}[thm]{Proposition}

\newtheorem{defn}[thm]{Definition}
\newtheorem{cor}[thm]{Corollary}

\begin{document}

\title[It's Not Fairness, and It's Not Fair]{It's Not Fairness, and It's Not Fair: The Failure of Distributional Equality and the Promise of Relational Equality in Complete-Information Hiring Games}

\author{Benjamin Fish}
\affiliation{
\institution{University of Michigan}
\city{}
\country{}
}
\email{benfish@umich.edu}

\author{Luke Stark}
\affiliation{
\institution{Western University}
\city{}
\country{}
}
\email{cstark23@uwo.ca}

\begin{abstract}
Existing efforts to formulate computational definitions of fairness have largely focused on distributional notions of equality, where equality is defined by the resources or decisions given to individuals in the system.  Yet existing discrimination and injustice is often the result of unequal social relations, rather than an unequal distribution of resources.  Here, we show how optimizing for existing computational and economic definitions of fairness and equality fail to prevent unequal social relations.  To do this, we provide an example of a self-confirming equilibrium in a simple hiring market that is relationally unequal but satisfies existing distributional notions of fairness.  In doing so, we introduce a notion of blatant relational unfairness for complete-information games, and discuss how this definition helps initiate a new approach to incorporating relational equality into computational systems. 
\end{abstract}

\maketitle

\section{Introduction}
Scholars, designers, and engineers continue to be confronted with a thorny challenge: how to approach the development of technologies like machine learning in practice in ways that account for social norms, ethics, and human values \cite{Eubanks:2018wv,Hutchinson:2019ce,Lee:2018ik,Woodruff:2018if}. Many algorithmic techniques attempting to encode values like fairness have been developed in response \cite{barocas-hardt-narayanan,Verma:2018hw,DworkR14}. Though popular, recent scholarship has described some of the methodological and practical challenges in deploying such tools in the pursuit of ameliorating social inequality and injustice \cite{JafariNaimi:2015en,Eubanks:2018wv,Hoffmann:2019hr,SelbstBFVV19, Benthall:2019dpa}. Moreover, the turn to computational definitions of values like fairness as mechanisms to solve problems of social injustice presents problems in itself \cite{Agre:1997ts,Friedman:1996tm,SelbstBFVV19,Hoffmann:2019hr,Passi:2019ey}. Algorithms and their makers often do not just ignore, but at times exacerbate the material conditions under which inequality flourishes \cite{Eubanks:2018wv,Hoffmann:2019hr}. For instance, algorithms often lack what 
has been elsewhere termed \emph{value legibility} \cite{FishS21}, the degree to which broader contexts or social processes that are known to have normative impacts are left out of a model's design, and/or are unaccounted for in its deployment.  

One core feature of these ongoing debates involves the narrow understanding of fairness, a central value in the field, as a concept. The vast majority of extant fairness models as applied to real-world problems are grounded in a constrained set of philosophical definitions of what fairness is, and what it is for \cite{Binns:2018tf,Heidari:2018us}. These definitions center on what egalitarian philosophers term \emph{distributional equality}, or how material goods ought to be allocated fairly. Yet these theories are poor mechanisms for understanding how to promote equality of social relations between human beings and the elimination of social hierarchies of status and power \cite{Anderson:1999wh}: this latter kind of fairness, known as \emph{relational equality}, is largely missing from conversations around algorithmic fairness as a guiding ethical theory.  

If we wish both to improve the design of computational systems, including in uses for progressive social ends, we need to be able to compare strategies for decision-making along the axis of relational equality. Relational equality cannot be accomplished through formal modeling alone \cite{FishS21}, but developing computational mechanisms for representing this form of fairness is a necessary though not sufficient condition of relational equality's broader project. 

Here, we show how existing computational and economic definitions of fairness (and discrimination) fail to capture relational (in-)equality via a simple game-theoretic model of a hiring market. Not only are hiring markets and algorithmic hiring a ripe venue for concerns of discrimination \cite{RaghavanBKL20,bogen2018help}, but they have long been a focus for understanding discrimination, particularly statistical discrimination \cite{Arrow:1971tv,Phelps:1972tu,fang2011theories}.

We provide \emph{self-confirming equilibria} in this hiring market where a firm will treat a candidate poorly (by only offering them low wages) if the firm believes other firms will also treat the candidate poorly.  We claim that these equilibria are discriminatory and define them as \emph{blatantly} relationally unfair.  Doing so allows us to overcome the difficulties of translating relational equality into a formal definition:  while relational equality may be a well-formed ethical theory, it is not particularly amenable to computational formalization.  We then show how satisfying existing computational and economic definitions of fairness fails to prevent these equilibria, and thus fails to prevent relational inequality, without having to formally define relational fairness---merely a form of blatant relational unfairness.\footnote{This approach also has the benefit that we need not believe that relational equality should be an ethical end to conclude that existing definitions of fairness are insufficient:  all we need to believe is that this particular situation in this hiring market is in some way unfair or unjust.  This contrasts with much previous work where the definition of computational fairness used is never justified but simply declared as axiomatic.}  

In particular, we demonstrate this discrimination in our model hiring market is neither taste-based nor statistical: the firm need not hold any personal animus against the candidate, and the candidate may not be any less qualified or skilled than other candidates treated more advantageously, even at equilibrium.  We also show this type of equilibrium is not prevented by enforcing extant computational definitions of fairness: neither by enforcing forms of group fairness, like equalized odds \cite{Hardt:2016wv} or statistical parity \cite{kamiran2009classifying}, nor individual notions of fairness, like Dworkian individual fairness \cite{Dwork:2011vl}, nor even causal notions of fairness \cite{KusnerLRS17}.

While the equilibria that result in this kind of discrimination are relatively easy to find, to the best of our knowledge the implications of such self-confirming equilibria for the formal treatment of inequality have not been made clear prior to this work, nor recognized as a distinct form of discrimination.  These implications are significant.  That existing computational definitions of fairness fail to prevent the relational inequality we demonstrate implies a need for novel treatments of discrimination and unfairness which center relational equality, rather than focusing only on distributional equality.  We initiate this work by providing a formal definition of \emph{blatant unfairness} in complete-information games.  This definition helps us formalize the most egregious kinds of relational inequality, while still capturing the close connection between distributional inequality and relational inequality.  Given the ubiquity of unequal social hierarchies not only in situations of material redistribution (such as hiring markets) but also more broadly across all kinds of social relations, our formal definition of blatant unfairness seeks to start making legible a wider array of forms of unfairness than existing computational solutions.  Complete-information games are simple enough to largely ignore the most thorny issues of relationships -- players already agree on their payoffs and the rules of the game -- and thus represent a convenient starting point for the study of computational relational equality.

We first provide background on theories of relational equality, extant computational definitions of fairness, and statistical discrimination and game theory (Section \ref{sec:background}). We then introduce our model of a hiring market and our definition of blatant unfairness, grounding this definition in theories of relational equality, and show how existing definitions fail to prevent blatant unfairness (Section \ref{sec:blatant_unfairness}) and then extend our work to games with more than two players (Section \ref{sec:multi-player}).  We conclude with a short discussion examining the limitations of and future directions for this work (Section \ref{sec:discussion}).

\section{Defining Algorithmic Fairness}\label{sec:background}

Appeals to consider the broader societal implications of algorithmic systems are often implicitly calls to focus on issues of relational equality, and existing debates around algorithmic fairness all touch on aspects of relational equality. However, none captures the full bundle of related concerns relational equality makes conceptually tractable.  Perhaps the most closely related work is that of Kasy and Abebe \cite{KasyA21}, who recognize limitations in existing definitions of fairness closely related to the broader concerns relational equality brings to the fore.  Other related works include that of Birhane \cite{Birhane21}, who focuses on issues of Afro-feminist relational ethics in algorithmic systems, and Viljoen \cite{viljoen2020democratic}, who proposes a relational account of data and its governance.  

In contrast to Kasy and Abebe \cite{KasyA21}, we argue that the necessity of widening the scope of what fairness means, especially in order to include theories of relational equality, requires entirely rethinking existing formal models of welfare and fairness altogether. We build on wider debates around the differences between what have been variously described as distributional versus dignitary \cite{Hoffmann:2019hr} or allocative versus representational harms \cite{AbbasiFSV19}. There is now a robust tradition in algorithmic fairness arguing broadly for attention to social inequalities through critiques of algorithmic distribution in a wide set of areas \cite{Barabas:2018vc, Eubanks:2018wv, CostanzaChock:2018fc, Anonymous:4jyXmvug, Dencik:2019jj,Park:2019fk}.  Focusing on relational equality as a concept draws together many of these extant criticisms of algorithmic fairness under one banner: that technologists must attend to the realities of society's inequalities and injustices whatever else they do. 

Theories of relational equality emphasize the centrality of fair social relations and the elimination of status and power hierarchies as the ``point'' of egalitarian thought. Relational equality provides an overarching frame for the wide variety of critiques leveled at algorithmic fairness work from critics both internal and external to the community \cite{Walzer:1983wu, Schemmel:2011ic, Scheffler:2013vx}. Our focus here is on the variant of relational equality proposed by Elizabeth Anderson \cite{Anderson:1999wh} known as democratic equality. In her initial 1999 article describing democratic equality and in subsequent work \cite{Anderson:2007uo,Anderson:2008gw,Anderson:2013hu,Anderson:2015wq}, Anderson has pointed out that the proponents of distributional equality in political theory \cite{Cohen:1989td,Dworkin:1981vu} are often more inadvertently damning of the concept of equality than supportive of it. For Anderson, many extant theories of equality combine the worst abuses of free market and centralized planning systems: such theories both neglect the social equality of large portions of the population, such as stay-at-home caregivers, but also entail invasive, paternalistic oversight by the state in order to determine moral desert. 

Instead, ``the proper negative aim of egalitarian justice is not to eliminate the impact of brute luck from human affairs, but to end oppression, which by definition is socially imposed''; likewise, egalitarian justice's ``proper positive aim is not to ensure that everyone gets what they morally deserve, but to create a community in which people stand in relations of equality to others'' \cite{Anderson:1999wh} (288).  Equality is a question of social relations first and foremost, grounded in the ``equal moral worth of persons'' (312).  Differences in the diversity of human experience do not justify unequal social relations.

Anderson's democratic equality is notable for being resolutely opposed to paternalistic coercion and moral judgments by institutions, a practice simultaneously under fire in debates around digital ``nudging'' \cite{gandy2019toward}. In the context of algorithmic fairness, democratic equality thus also be understood as a critique of what Zuboff \cite{Zuboff:2019uf} terms ``surveillance capitalism'' facilitated by widespread digital surveillance, data collection, and ML-based analytics. These systems only enhance the capabilities of what Anderson elsewhere describes as ``private government" \cite{Anderson:2015wq} -- the ability of institutions such as private enterprises to command and control their workers while enabling limited recourse. Given that public authorities sometimes also practice such coercion as well, Anderson observes that democratic equality, with its emphasis on individual dignity, lack of social standing, and democratically enabled decisions, is a mechanism to ensure the state remains fair, accountable, and benefits all equally \cite{Anderson:2008gw}. Nonetheless, relational equality takes as axiomatic that, to paraphrase the poet John Donne, "no one is an island": equal social relations are in some way dependent on material equality and vice versa.  By bringing together critiques of existing conceptions of equality focused solely on material equality with a shift in focus to social relations, Anderson's conception of relational equality provides a potential path forward to overcome the call in algorithmic settings to attend to the realities of society's injustices and unequal social relations.

\subsection{Group Fairness}

Statistical notions of fairness based around distributional equality are central to current debates around discrimination and inequality in machine learning.  Many of the most popular such notions are those that fall under the penumbra of what we refer to as \emph{group fairness}.  In contrast to relational equality, these computational models of fairness center around a particular understanding of how to distribute material resources between groups known as equality of opportunity (EOP). As Heidari et al. \cite{Heidari:2018us} have noted, notions of group fairness parallel extant variations of EOP from political philosophy. The authors show how specific fairness definitions employing predictive value parity \cite{KleinbergMR17} and equality of odds \cite{Hardt:2016wv} map to particular variants of EOP: the former to Rawlsian ``fair'' EOP \cite{Rawls:2009ux}, the latter to so called ``luck'' egalitarianism \cite{Roemer:2002gx}. Broadly speaking, luck egalitarianism is premised on the moral necessity of only equalizing inequalities based on luck, or circumstances outside an individuals control. As Binns \cite{Binns:2018tf} notes, notions such as relative welfare (understood in terms of pleasure or preference-satisfaction) \cite{Cohen:1989td,Dworkin:1981vu}; material resources \cite{Rawls:2009ux,Dworkin:1981ux}, or capabilities (the ability and resources necessary to act) \cite{Sen:1992ue} can all be equalized via EOP approaches: all these models of equality center on determining the proper distribution of the good or resource in question.   

In machine learning, group fairness metrics are typically defined for supervised learning settings using a sensitive feature of the data to be analyzed.  Here, we consider a feature space $X$ where each individual has feature values $x\in X$, and a \emph{sensitive feature} $a\in A$, which represents a demographic feature that we are concerned people should not be held responsible for, such as a race or gender.  Each individual also has a target label $y\in Y$; the goal is to construct a classifier $g:X\rightarrow D$ for some decision space $D\subset\mathbb{R}$.  When otherwise unclear, we use $A(x)$ to denote the sensitive attribute of individual $x$ and $Y(x)$ the label of $x$.  We now introduce a notion of exact group fairness that generalizes many of the well-established types of group fairness metrics.  As in typical supervised machine learning, group fairness metrics assume a distribution over individuals; for the definition of group fairness we will abuse notation and also use $X$ to denote a random variable over individuals, and similarly with $Y$ and $A$.
\begin{defn}
A classifier $g:X\rightarrow D$ satisfies exact $(F_1,F_2)$-\emph{group fairness} with respect to a sensitive attribute $A$ for measurable functions $F_1$ and $F_2$ if
\[F_1(g(X),Y) \perp A \mid F_2(g(X),Y).\]
\end{defn}

Group fairness parameterizes a wide variety of popular definitions based on equalizing statistics from a confusion matrix.    For example, if $F_1(g(x),Y(x)) = g(x)$ and $F_2$ is any constant, then $(F_1,F_2)$-group fairness is statistical parity:

\begin{defn}[Kamiran and Calders \cite{kamiran2009classifying}, Dwork et al. \cite{Dwork:2011vl}]
A classifier $g:X\rightarrow\mathbb{R}$ satisfies \emph{statistical parity} if, given a distribution over $X$ and a sensitive attribute $A$,
\[P[g(x) \ge z | A =a] = P[g(x) \ge z]\]
for all $a\in A$ and $z\in\mathbb{R}$.
\end{defn}

Group fairness also includes other popular definitions, like \emph{equalized odds} \cite{Hardt:2016wv}, where $F_1(g(x),Y(x)) = g(x)$ and $F_2(g(x),Y(x)) = Y(x)$.

While group fairness definitions have become increasingly popular in the computational literature \cite{barocas-hardt-narayanan,du2020fairness,Verma:2018hw}, scholars have posited several criticisms of these definitions, ranging from an overly narrow conceptualization of disadvantage \cite{Hoffmann:2019hr} to how they fail to treat similarly ``risky'' people similarly \cite{CorbettDavies:2018ud}.  In response to these perceived failures, several alternative definitions have been proposed.

\subsection{Individual and Causal Fairness}

Group fairness also fails to prevent discrimination against individuals, provoking Dwork et al. \cite{Dwork:2011vl} to introduce \emph{individual fairness}.

\begin{defn}[Dwork et al. \cite{Dwork:2011vl}]
A classifier $g:X\rightarrow D$ satisfies $(M,m)$-\emph{individual fairness} if for every $x,y\in X$, $M(g(x),g(y)) \le m(x,y)$, where $M$ is a statistical distance (we assume $M(g(x),g(x))=0$) and $m$ is a metric.
\end{defn}

Group fairness attempts to ensure equality across people under the intuition that people with one value of a sensitive feature are as deserving as people with another value of a sensitive feature.  But changing the sensitive feature may change other features of an individual, which requires examining counterfactuals. This reality has prompted increased attention to \emph{causal fairness}.

We focus on the strongest notion of causal fairness in the literature; if satisfying this notion of causal fairness still appears relationally unfair, then so will satisfying weaker variants.  This notion is also an individual notion of fairness, like Dworkian individual fairness, but instead of comparing individuals to each other, causal fairness compares an individual to the counterfactual where they had a different value of the sensitive attribute.

\begin{defn}[Kusner et al. \cite{KusnerLRS17}]
A classifier $g:X\rightarrow D$ is \emph{individually counterfactually fair} with respect to a causal graph if
\[P[g(x,a) = d|X=x,A=a] = P[g(x,a') = d|X=x,A=a]\]
for all $x\in X$, $a,a'\in A$, and $d\in D$.
\end{defn}

For an introduction to causal models and graphs, see Pearl \cite{pearl2009causal}.

\subsection{Criticisms of computational definitions of fairness}
Prior to this work, there have been a variety of other critiques of existing computational fairness definitions.  We briefly review these critiques, as well as how our critique put forth in this work differs from these.

Group fairness in particular has seen much criticism ranging from an overly narrow conceptualization of disadvantage \cite{Hoffmann:2019hr} to how they fail to treat similarly ``risky'' people similarly \cite{CorbettDavies:2018ud}.  There has also been criticism specific to particular definitions of group fairness.  For example, Dwork et al. \cite{Dwork:2011vl} and Hardt and Srebro \cite{Hardt:2016wv} criticize statistical parity for reducing utility to the entity performing the classification, as satisfying statistical parity may not permit a classifier that agrees with the target labels $Y$.  To deal with this issue, Hardt and Srebro introduce equalized odds (where $F_1(g(x),Y(x)) = g(x)$ and $F_2(g(x),Y(x)) = Y(x)$).  However, equalized odds, as a notion of group fairness, still does not escape the criticisms that are applicable to all group fairness definitions.

Our critique of group fairness is an extension of those focused on its normative and philosophical underpinnings (e.g. Hoffman \cite{Hoffmann:2019hr}).  More directly, we can apply here the same criticism that Anderson levels against proponents of distributional equality in political theory, since the existing definitions of fairness largely implement these exact theories of distributional equality \cite{Heidari:2018us}. In asking for an equality of resources, existing formal fairness definitions neglect the social equality of workers for which firms and employers believe have poor outside options, such as ``no-wage'' industries like stay-at-home caregiving \cite{Doucet_2015} or low-wage industries like fast food or package delivery \cite{Savas_2010}.  Yet existing definitions also involve a paternalistic oversight by decision makers.  They get to decide who deserves equal treatment and who does not, via control over the inputs and outputs of algorithmic fairness:  e.g. what data is used to create labels, what counts as a sensitive feature, and even who gets to skip the algorithmic decision making system entirely \cite{Scheuerman_Denton_Hanna_2021,Denton_Hanna_Amironesei_Smart_Nicole_2021}.

We demonstrate how an ethical theory grounded in relational equality renders group fairness undesirable to use:  if everybody is treated equally but equally badly, then this outcome will appear unfair but still satisfies group fairness. In contrast to previous technical work where such criticisms, at least in principal, might be avoided by minor changes to the definition, Anderson's critiques applies much more broadly -- and not only for group fairness definitions, but also Dworkian individual fairness and causal fairness.  

Previous criticism of individual fairness has largely focused on the practical difficulties in constructing the metric $m$, which encodes who should be treated similarly to each other and who should not \cite{GillenJKR18,JungKNRSS19,Ilvento20}.  This criticism closely mirrors the debates within theories of EOP as to when people should be held responsible for differences between themselves. While we agree that group fairness fundamentally fails to protect individuals, it may be the case that constructing the metric is a solvable problem, subject to a particular interpretation of EOP. However, our criticism focuses on the fact that individual fairness is still a notion of distributional equality:  just as in group fairness, if everybody is treated equally but equally badly, then this outcome seems unfair but still satisfies Dworkian individual fairness.

Meanwhile, critiques of causal notions of fairness have focused on two basic conceptual issues around using causal models.  The first kind of critique centers on the difficulties in practice of manipulating sensitive features:  an individual cannot practically change the racial category they are assigned, especially at classification time \cite{kohler2018eddie}.  The second kind of critique focuses on the difficulties of using sensitive features with complex social meanings as stable entities that can have causal effects.  The concern is that what may seem like an effect of the sensitive feature actually constitutes its social meaning \cite{hu2020s,kohler2018eddie}.  Moreover, the only way we can experimentally derive causal effects of sensitive features is negatively, by fixing a selected set of other features and showing that they had no effect on the target of interest by fixing them, and that it must therefore have been because of the change in value of the sensitive feature (as in r\'esum\'e audit studies \cite{bertrand2004emily}).  Similarly to individual fairness, where the metric provides the substantive guidance on what counts as fairness, here the features we chose to fix provide what should count as fair and what should not \cite{hu2020direct}.

These criticisms rely on the particulars of the sensitive feature mobilized as a social category, and in the difficulties of constructing a causal model that represents such categories.  Here, however, we avoid the difficulties of formalizing this kind of criticism by focusing instead on the simple fact that causal fairness is still distributional, in the sense that it seeks to equalize the distribution of resources (though of course not ``distributional'' in the sense of only using the joint statistics of the features and classifier).

Finally, a line of previous work implicitly criticizes some of these definitions of fairness by showing that placing a fairness constraint on the classifier can still lead to unfairness in the overall system via feedback loops, and that therefore ensuring fairness requires examining the entire system \cite{ElzaynJJKNRS19,LiuDRSH19,MouzannarOS19,DAmourSABSH20}.  But this work demonstrates only that it is ineffective to place a constraint only on the classifier in order to ensure a given fairness definition holds over time.  This work does not demonstrate that the definition of fairness should not be solely a function of the classifier to begin with.  In contrast, we provide a concrete model where unfairness cannot be defined solely as a function of a classifier and should be defined as a property of the entire system instead.

\subsection{Statistical discrimination, markets, and equilibria}

In this work, we are primarily concerned with classifiers being played as strategies in games, and in particular, games featuring surpluses generated between a firm and a candidate for a job.  We now introduce some necessary basics from game theory, along with the relevant existing notions of fairness from the economics literature.  In games, and in games representing hiring markets in particular, the economics literature focuses on discrimination arising not just from prediction error, as in group and individual fairness, but from behavioral biases.  The two most common forms of this are taste-based discrimination and statistical discrimination.  Under taste-based discrimination, discriminatory outcomes are produced by the direct animus of an individual or group.  Setting aside the difficulties in making causal models, animus does not need mediation by other factors to link a sensitive feature to an outcome.  So even if it is hard to formally define ``direct animus'', we can avoid this issue by saying that at the very least, if we are given a causal graph and there's no edge from $A$ to $D$, there can be no taste-based discrimination:

\begin{defn}
A classifier $g:X\rightarrow D$ exhibits no \emph{taste-based discrimination} with respect to a sensitive feature $A$ in a causal graph if there is no directed edge from $A$ to $D$.
\end{defn}

Group-level disparity can have other sources, however.  \emph{Statistical discrimination} is when disparate group outcomes result from the actions of utility-maximizing agents that use observable characteristics of the groups to infer outcome-relevant characteristics of individuals, first described by Arrow \cite{Arrow:1971tv} and Phelps \cite{Phelps:1972tu}.  Disparate outcomes under statistical discrimination, are often, though not always, the result of informational frictions, like an unobserved skill level of workers \cite{fang2011theories}.  If a firm observes group membership, and group membership is correlated to skill level, then the firm will use group membership as a proxy to hire workers, thereby disincentivizing workers of one group to invest in skills at equilibrium which in turn creates the correlation between group membership and skill level.

Crucially, in models of statistical discrimination, the decision maker claims to be ``justified'' in their use of group membership, defined by some sensitive feature $A$, in the sense that there is a real difference, at equilibrium, between the groups:  one group will on average perform worse than the other, and thus generate less \emph{surplus}, i.e. the total benefits reaped from the decision maker hiring the workers. Statistical discrimination externalizes the reasons for differences between groups, such as historical and/or ongoing discrimination and animus. 

Here we use the amount of surplus generated as the skill level.  If the firm believes one group will on average generate less surplus than the other, then given a noisy estimate of the surplus of a worker from that group, the firm will use a higher threshold on the noisy estimate to hire them than for the other group.  Different thresholds mean the worker with the \emph{least} surplus hired will be different across group membership.  Using this difference of marginal productivity as a definition of bias is the so-called Becker test \cite{cowgill2019economics}.  Regardless of whether the firm fails the Becker test or not, though, statistical discrimination is when there is a true average difference of surplus across $A$, and the firm uses $A$ to take advantage of this, creating different outcomes across $A$.  At the very least, then: 
\begin{defn}
A classifier played by a decision maker $f$ exhibits no \emph{statistical discrimination} with respect to a sensitive feature $A$ if the average surplus generated between the decision maker and a player is the same across all groups defined by $A$ at equilibrium.
\end{defn}

That the decision maker needs to generalize using the sensitive feature to maximize utility may make it difficult to argue that statistical discrimination is always undesirable or unfair -- if statistical discrimination is always wrong, then it would preclude all kinds of generalization, even against innocuous features (say, being a smoker, for public health reasons) \cite{Binns:2018tf}.  However, personalization thus enabled may also entail implicit judgements regarding moral desert, creating the kinds of paternalistic intervention Anderson's relational equality wishes to avoid.  Regardless, this kind of criticism of statistical discrimination does not preclude a definition of unfairness based on the Becker test, which could include some subset of the combination of taste-based and statistical discrimination.  In contrast, we will claim in this work that there are blatantly unfair situations that passes the Becker test where the discrimination at issue is neither taste-based nor statistical.
Besides taste-based and statistical discrimination, there have been a few other approaches, including network-based discrimination \cite{arrow1998has}, though these also largely feature true productivity differences between groups or individuals (though, see \cite{Fish:2019gi}).

In order to find situations that are unfair, yet display no statistical or taste-based discrimination and would pass the Becker test requires situations in which a firm believes in some kind of difference across groups, even when there isn't any.  This requires introducing a solution concept where beliefs are not always correct, which means we can't use Nash equilibria. Here we consider extensive-form games with perfect recall and complete information, where the rules of the game and the payoffs are common knowledge, and players remember the actions they play.  Let's start by rephrasing the definition of a Nash equilibrium to see how it is the result of players maximizing utility according to their beliefs, such that their beliefs are always correct.  To do this, we introduce some notation and recall some necessary basics from game theory.

We will denote $s$ as a pure strategy profile, so $s_i$ is a pure strategy for player $i$, which maps each node in the extensive-form game to an action, and $\pi$ as a mixed strategy profile, meaning $\pi_i$ is a distribution over the possible pure strategies for player $i$.
Recall that an information set $h\in H$ in an extensive-form game is a partition of the non-terminal nodes.  Call $H_i$ the information sets where player $i$ has the move (so $H_{-i}$ is the remaining information sets) and $H(\pi)$ to be the set of information sets that can be reached under mixed strategy profile $\pi$ with positive probability.  Since we only consider games with perfect recall, every mixed strategy has an equivalent behavior strategy with an identical distribution over outcomes.  Recall a behavior strategy for player $i$ maps each information set where $i$ plays to a distribution over actions.   We will use $\sigma$ to refer to a behavior strategy profile, so that $\sigma_i$ is the behavior strategy for $i$, $\sigma_{-i}$ is the remaining behavior strategies, and $\sigma_i(h_i)$ is a distribution over actions at information set $h_i\in H_i$.  Since any mixed strategy has an equivalent behavior strategy, given a mixed strategy $\pi_i$, we can denote $\hat{\pi}_i$ as the equivalent behavior strategy induced by $\pi_i$.  A player's belief should represent $i$'s conception of the other player's strategies,  so we define a belief $\mu_i$ to be a distribution over profiles of behavior strategies for each of $i$'s opponents.  These behavior strategies need not a priori be strategies actually played by anyone.  Finally, let $u_i$ be the utility function for player $i$ and let $u_i(s_i,\mu_i)$ be the expected utility for player $i$ if $i$ plays pure strategy $s_i$ and the remaining players play a random strategy drawn from belief $\mu_i$.  As a warmup, with this notation, a Nash equilibrium is a strategy where each player maximizes their utility according to their beliefs, such that their beliefs about what their opponents will play are correct with probability $1$.
 
\begin{defn}[see Fudenberg and Levine \cite{fudenberg1993self}]
A strategy profile $\pi$ is a \emph{Nash equilibrium} if for each player $i$ and and each strategy $s_i$ with $\pi_i(s_i) > 0$, there are beliefs $\mu_i$ such that
\begin{itemize}
\item $s_i$ maximizes $u_i(\cdot,\mu_i)$ and
\item $P_{\sigma_{-i}\sim \mu_i}[\sigma_j(h_j) = \hat{\pi}_j(h_j)] =1$ for all $j\neq i$ and $h_j\in H_{-i}$.
\end{itemize}
\end{defn}
The second constraint, which must hold over all $h_j\in H_{-i}$, ensures that $\mu_i$ is always correct.  In this work, we employ self-confirming equilibria, where each player $i$'s beliefs need not be correct at all of their opponent's information sets, merely the information sets ${H}(s_i,\pi_{-i})$ that actually happen at equilibrium; players may be arbitrarily wrong about strategies on off-equilibrium paths.

\begin{defn}[Fudenberg and Levine \cite{fudenberg1993self}]
A strategy profile $\pi$ is a \emph{self-confirming equilibrium} (SCE) if for each player $i$ and and each strategy $s_i$ with $\pi_i(s_i) > 0$, there are beliefs $\mu_i$ such that
\begin{itemize}
\item $s_i$ maximizes $u_i(\cdot,\mu_i)$ and
\item $P_{\sigma_{-i}\sim \mu_i}[\sigma_j(h_j) = \hat{\pi}_j(h_j)] =1$ for all $j\neq i$ and $h_j\in {H}(s_i,\pi_{-i})$.
\end{itemize}
\end{defn}
This definition allows the beliefs to vary depending on the strategy $s_i$, which makes sense when a single ``player'' consists of a distribution over different agents with different beliefs but the same utility function, each playing a pure strategy.  A \emph{unitary} self-confirming equilibrium is when the belief $\mu_i$ can only depend on the player and not $s_i$.  In this work, all self-confirming equilibria will be unitary, but because this difference is not important for us, we will not discuss it further.

\section{The job market and blatant unfairness in two-player games}\label{sec:blatant_unfairness}

\subsection{An example}\label{sec:example}

Interrogating social inequality in the workplace is central to Anderson's notion of relational equality \cite{Anderson:2015wq}. As such, we begin by introducing a central example focused on hiring, which will help us define blatant unfairness formally, and show that our notion is incompatible, a la Chouldechova et al.\ \cite{Chouldechova17} and Kleinberg et al.\ \cite{KleinbergMR17}, with existing computational definitions of fairness.  This example also serves as a thought experiment, in the philosophical sense, to build our moral intuition as a way to motivate an alternative ethical, philosophical, and mathematical approach to those approaches used by extant mathematical definitions of fairness.

In this example, an individual is applying for a job at a firm, in a market sufficiently large and important that there is no realistic alternative for the individual to seek other markets or other alternatives.  Consider the scenario -- \emph{scenario A} -- where the firm and the job candidate believe the job candidate has few options to work elsewhere in this market, or only at lower wages. The firm is thus highly incentivized neither to make the job easy to get nor offer the candidate high wages.  Yet even when it is challenging for the job candidate to get the job, they will try anyway; and if they are offered low wages, they will accept it.  Moreover, the firm will indeed offer low wages, as they acquire a worker without having to pay them a premium.  

Now consider an alternative scenario -- \emph{scenario B} -- also featuring a candidate for the same job.  (The nature of this alternative does not matter much right now as a matter of normative relevance:  it may be a counterfactual, or simply occur at a different time or place, or even consist of a different firm under the same incentives.)  In this scenario, the candidate is exactly as qualified for the position as the candidate in the previous scenario, but now both the firm and this candidate believe that this candidate has many options to work elsewhere.  Now the firm is incentivized to offer the candidate a job quickly, and with higher wages.

Neither of these scenarios are meant to be exact depictions of realistic hiring markets.  It would be extremely difficult to find a natural experiment in which the only difference is in the players' beliefs, for example.  Rather, these examples are meant as a thought experiment to demonstrate a particular example of an unequal social relationship, which is an important feature of Andersonian democratic inequality.

The social relationship at question is the one between the firm and the job candidate, which is enacted via the actions the firm and job candidate take while bargaining over wages: the firm, what wage to offer, if any, and the candidate, what wages they're willing to accept.  These actions affect both the firm and the candidate and are shaped by their beliefs about each other: in other words, it is a social relationship.  But it is not an equal one.  A candidate stuck in scenario A will always have a difficult time getting a job in this market: this candidate is treated as ``lesser than'' the candidate benefitting from scenario B, even if the firm knows that they would both still benefit by moving to scenario B.  The candidate from scenario A is never treated as worthwhile by the firm, even when ``worth'' is defined narrowly as the worth of the candidate to the firm.

Nor does it appear that it is the candidate's fault in scenario A, because we assume the candidate in scenario A did not differ from the candidate in scenario B in any meaningful way.  It is the case that if the candidate's belief is not true that they have few options, they should instead simply change their beliefs.  But because the candidate from scenario A doesn't believe that they will get a job elsewhere, the candidate always accepts the bad offer from the firm, and so they never get a chance to explore their outside options; thus the candidate never updates their priors.  

Worse, the scenario above can involve a self-fulfilling prophecy of another kind.  Even if the unfortunate candidate is wrong in their priors initially, they will soon be proven right in their beliefs if they always face scenario A. If each and every firm believes this unfortunate candidate has no outside option, then the entire market will treat the candidate in one way, and the belief of the market will become reality: the unfortunate truly will not have an outside option. If everyone believes it, it becomes true. And if this hiring market is sufficiently large, or firms' beliefs are uniform across different hiring markets (motivating an inelastic demand for jobs), the ability to move to a different hiring market may not exist at all, or be impractical even if it does in theory.  This scenario makes it hard to argue that the candidate from scenario A doesn't deserve or need a job in this particular market, and should find employment in some greener pasture where firms have different beliefs about social dessert.

Thus, we can conclude that scenario A is an example of a lack of social standing, and that under Anderson's democratic equality, this is a blatantly relationally unfair outcome.  Admittedly, there is no ``hard evidence'' that this relationship is unfair, nor can there be.  Instead, this sense of unfairness can come to us as either a particular interpretation of Andersonian relational equality or simply a moral intuition that it is intolerable when individuals are treated as ``lesser than''.  In more realistic settings, this social relationship extends well beyond the single point of bargaining which we focus on here, extending both before to their relationship in the market and after to their relationship as employer and employee.  Likewise, Andersonian relational equality also extends well beyond this point, by considering how democratically enabled decisions, individual dignity, etc., enable relational equality on a much wider scale.

However, this simplified relationship consisting of bargaining over wages is still sufficient to find an example, summarized in scenario A, of unequal social standing, even if it is not sufficient to characterize the full scope of relational equality.  This will allow us to create a concrete model of blatant relational inequality by modeling these scenarios as simple games.

\subsection{The market}
We consider a non-repeated Nash bargaining game with two players, the firm and the candidate, splitting a unit surplus, plus a third player, the market, who provides the disagreement point, i.e.\ an outside option.  The surplus represents the total benefits reaped from the candidate working the job the firm is offering.  We represent each candidate as a vector $x\in X$, and denote the firm by $f$.  We assume that the feature vector uniquely identifies the candidate, so the game remains a complete-information one.

Both players have need:  the firm needs to fill the job and the candidate needs a job, or else suffers a negative utility.  For simplicity, we set this negative utility for each of them as $-1$, and since they produce unit surplus together,  the firm's action is to make an offer $\pi_f\in[0,3]$, representing their wages:  salary, benefits, etc.  Then the candidate just chooses whether to accept or reject this offer.  If the candidate accepts, the firm receives $2-\pi_f$ utility and the candidate receives $\pi_f-1$.  For example, $\pi_f=1$ represents the break-even point for the candidate and $\pi_f=3/2$ represents an even split of the surplus.  If the candidate rejects, the market player $m$ takes an action $(o(f),o(x))\in[0,3]^2$, and the firm receives its outside option $2-o(f)$ and the candidate receives its outside option $o(x)-1$.  Given that in practice the market player $m$ isn't a single agent, but rather itself is the result of another, possibly equilibrium, outcome, we do not explicitly assign payoffs to this player.  Thinking of $m$ as the ``Nature'' player, the market plays a fixed strategy.  We denote this game $M_{x,f}$.

We let the firm $f$ use a classifier $g:X\rightarrow \Delta([0,3])$ mapping candidates to probability distributions over offers to decide on a strategy $\pi_f=g(x)$.  We use $\pi_f$ and $g(x)$ interchangeably for this reason.  In this work, we do not consider what the features are, or even how the classifier is learned.

Thus $g$ represents the strategy of $f$ for a collection of games, each game $M_{x,f}$ parameterized by $x$.  If $g(x)$ is an action played at equilibrium in every game, we call the classifier $g$ an equilibrium strategy.

There are infinitely many Nash equilibria in $M_{x,f}$.  For example, there is an equilibrium where the firm offers $g(x)=o(f)=o(x)$.  (We give a slight generalization of this fact in Proposition \ref{prop:market_sce}.)  Here, the firm offers market-rate wages, and leaves the candidate indifferent to whether or not to accept the offer.  However, simply the fact of many equilibria doesn't fully capture scenarios $A$ and $B$.  While there are equilibria that features both low and high values of $\pi_f$, this requires the market play different values $o(x)$ depending on the player.  We don't want to assume this will happen a priori, because the surplus for each candidate is identical.  The scenarios described above are for the same market, meaning scenarios A and B wouldn't occur simultaneously.  Nash equilibria also do not take advantage of the impact of the players' beliefs.

Before turning to SCE, we will want simpler notation for a player's beliefs about what the outside options $o(f)$ and $o(x)$ are.  We will use $o_i(j)$ to denote what $i$ believes $j$'s \emph{expected} outside option is under a given belief $\mu_i$.
More formally, recall that a belief $\mu_i$ for any player $i$ is a distribution over behavior strategies for each opponent.  In our case, beliefs will not correlate $m$'s strategy with any other strategies, and moreover, $m$ only plays at one information set, so abusing notation slightly, we can call $\mu_i(m)$ $i$'s belief about what $m$ will play, i.e. a distribution over outside options $o\in[0,3]^2$ indexed by $f$ and $x$.  Then let $o_i(j) = E_{o\sim\mu_i(m)}[o(j)]$, $i$'s belief about $j$'s expected outside option, $i,j\in X\cup\{f\}$.
\subsection{Blatant unfairness in the market}

Particular SCE in $M_{x,f}$ will turn out to exactly be the unfair behavior from Section \ref{sec:example}, so our first order of business is finding them.  Not only will equilibria in these games form a concrete model for the scenarios in Section \ref{sec:example} and therefore for relational unfairness, but they will allow us to show a) this unfairness is not solely a function of a classifier, and therefore a definition of (un)fairness cannot be defined only as a function of a classifier, as much previous work does; and b) extract the relevant features of the undesirable behavior to be able to define a concrete definition of relational unfairness.

Importantly, we can find SCE where $o_x(f)$ and $o_f(f)$ are not only not equal to $o(f)$, they can be arbitrarily far apart from $o(f)$ (and likewise $o_x(x)$ and $o_f(x)$ can be arbitrarily far from $o(x)$):  the player's beliefs about their outside options are nowhere close to their true outside options.  This happens when the candidate accepts the firm's offer because the candidate believes their outside option is no higher than the firm's offer, and then neither of them will find out what the market would have done.  Moreover, the firm's offer may be arbitrarily low, as long as the firm believes they won't get any more from the market.  Since the market's action is now off the equilibrium path, neither the firm nor candidate need be correct about the market's strategy, justifying their beliefs.

\begin{prop}\label{prop:sce}
If beliefs satisfy $o_f(f) \ge o_f(x)$ and $o_f(x) \ge o_x(x)$, the strategy profile where, for all candidates $x$, the firm offers $g(x) = o_f(x)$ and the candidate accepts that offer (and acts arbitrarily on other offers) is an SCE witnessed by those beliefs.  Consequently, for any pure strategy $g:X\rightarrow [0,3]$ played by $f$, and any strategy by the market player, there exists a set of beliefs where $g$ is an SCE.

\end{prop}
Before moving on, it should be noted there is another SCE where the market never plays: If $o_f(f) < o_f(x)$ and $o_f(f) \ge o_x(x)$, then the strategy profile where the firm plays $g(x)=o_f(f)$ and the candidate accepts is an SCE.  This SCE, however, is unrealistic in the sense that a ``rational'' implication of the firm's own beliefs is the candidate shouldn't accept the offer (because $o_f(f) < o_f(x)$) but the firm also believes that they will accept the offer anyway (since an SCE requires the firm's beliefs to be true on information sets that are reached).

\begin{proof}[Proof of Proposition \ref{prop:sce}]
Suppose the firm plays $g(x) = o_f(x)$ and believes the candidate will accept if and only if $g(x) \ge o_f(x)$.  More formally, we set $\mu_f(x)$ to be a distribution with probability mass $1$ on the function that maps each node corresponding to an offer $z$ to the accept action if and only $z \ge o_f(x)$.  (Again, we're abusing notation slightly by referring to $f$'s beliefs about $x$'s behavior strategy as $\mu_f(x)$ without specifying a complete distribution over all player's behavior strategies.)  Further suppose the candidate correctly believes the firm will play $o_f(x)$, and accepts this offer.

The candidate believes they will receive $g(x)-1$ utility if they accept or $o_x(x)-1$ in expectation otherwise, so if $o_f(x) \ge o_x(x)$, the candidate believes accepting is a best response.  The firm believes they will receive $2-g(x)$ if the candidate accepts, and in that case, the maximum utility they believe they can get is $2-o_f(x)$ (since then $g(x) \ge o_f(x)$).  The firm also believes that if the candidate rejects the offer they will receive $2-o_f(f)$, so offering $o_f(x)$ is a best response when $o_f(f) \ge o_f(x)$.  Since the candidate accepts this offer, the only information sets reached are the firm's action and the candidate's action, and both of their actions occur with probability one under their own beliefs.

Thus, the above beliefs witness that these best responses form an SCE when $o_f(f) \ge o_f(x)$ and $o_f(x) \ge o_x(x)$.  Given an arbitrary deterministic classifier $g:X\rightarrow [0,3]$, for each $x$, there are always beliefs that satisfy these constraints: set the beliefs such that $g(x) = o_f(x) = o_x(x)$ and $o_f(f)=3$.  Then it immediately follows that for all $x$, the strategy profile where the firm plays $g(x)$ and the candidate accepts is an SCE.
\end{proof}

When $o_f(x)$ and $o_x(x)$ are small, then $g(x)$ will be as well:  this is scenario A.  The firm and the job candidate both believe the candidate has a small outside option, and the candidate accepts the low wages.  When $o_f(x)$ and $o_x(x)$ are large (and $o_f(f)$ and $o_f(x)$ even larger), then $g(x)$ will be large:  this is scenario B.  Both are SCE regardless of how large or small each candidate's actual outside option is.  With these SCE in mind, we can observe what happened in scenario A to make this outcome so undesirable. Scenario $A$ appears unfair to the candidate when $g(x)$ is sufficiently small, say $g(x) < 1$, and the candidate receives negative utility.  The candidate and the firm could have helped each other by splitting the surplus, but didn't.  The firm takes all of the surplus and still leaves the candidate with need.  It was not a result of any inherent property of the candidate, firm, or market conditions, since the outcome did not depend on the outside options $o(f),o(x)$.  There was no rational strategy the candidate could play to change this situation, since it is at equilibrium.  There was an alternative equilibrium, scenario B, where they did split the surplus, meaning scenario A wasn't the only outcome, even for self-interested players.  (We could also consider the equilibrium where $g(x) > 2$ and the firm receives negative utility, but we focus on $g(x) < 1$ since it is sometimes less clear if the firm should even count as a moral agent, as it may not represent any one individual.  It is also a symmetric situation, covered qualitatively by our results anyway.)  

It may be tempting to blame the firm for scenario A, as they failed to offer the candidate any of the surplus.  Yet to wholly blame the firm is a result of a normative framework grounded in EOP, rather than Andersonian relational equality.  The existing distributional notions of fairness grounded in EOP are, at least implicitly, blaming the firm by making fairness a property of the classifier used by the firm (and possibly the distribution over inputs to the classifier), and constraining it so that undesirable behavior cannot occur.  Such models attempt to ensure that the candidates get the outcome they ``deserve.''  Yet we now provide an SCE where the firm plays the exact same classifier in the exact same game, but the outcome no longer appears unfair at all (at least in the context of this game).  This not only makes it hard to place all of the blame on the firm, but it implies that unfairness is not just a property of the classifier, but of the entire market, including the players' beliefs.  In this SCE, which is also Nash, the difference is entirely in the beliefs, which are now correct about their outside options.  Whereas before, the rational candidate was forced into accepting an arbitrarily small offer regardless of their true outside options, now when the firm makes an offer below the candidate's outside option, the candidate rejects the offer and they both get market rates.  

\begin{prop}\label{prop:market_sce}
For any pure strategy $g:X\rightarrow [0,3]$ played by $f$, if beliefs satisfy $o(f) = o_f(f) = o_x(f) \le o(x) = o_f(x) = o_x(x)$, the strategy profile where the firm offers $g(x) \le o(f)$ and $x$ accepts if and only if $g(x) \ge o(x)$ is an SCE and a Nash equilibrium witnessed by those beliefs.
\end{prop}
\begin{proof}
Suppose the firm believes the candidate will accept if and only if $g(x) \ge o(x)$, and plays $g(x)$ such that $g(x) \le o(f)$.  Suppose the candidate correctly believes the firm will play $g(x)$.  The candidate's only two choices are to accept $g(x)$ or take their outside option $o_x(x)=o(x)$, so their best response is to accept if and only if $g(x) \ge o(x)$.  The firm, meanwhile, will never increase their utility by offering more than $o(f)$, as either the candidate rejects and they get exactly $2-o(f)$, or the candidate accepts and they receive strictly less.  Otherwise, they will receive exactly $2-o(f)$ since $o(f)\le o(x)$:  If $g(x)\le o(f) < o(x)$, then the candidate rejects and they receive their outside options, or similarly, if $g(x) < o(f) = o(x)$, again the offer is less than $x$'s outside option and they reject.  The only other case is when $g(x) = o(f) = o(x)$, and now the candidate accepts, but again the firm receives $2-g(x)=2-o(f)$.  So any offer $g(x) \le o(f)$ is a best response.  Moreover, by supposition, the candidate and firm's beliefs are correct on all information sets with probability one, so this is not only an SCE but a Nash equilibrium. 
\end{proof}
Suppose we are given a classifier where $g(x) \le o(f) < o(x)$ for all candidates, and $g(x)$ is sufficiently smaller than $o(x)$, say $g(x) < 1$ but $o(x)$ is at least $3/2$.  This is at equilibrium under Proposition \ref{prop:sce}, and in this SCE the candidate gets $g(x)$ and therefore negative utility, resulting in scenario A and what appears to be a blatantly unfair outcome for the candidate.  This same classifier is also at equilibrium under Proposition \ref{prop:market_sce}, but now the candidate gets the much larger $o(x)$, which is more similar to scenario B.  It's not necessarily clear if we should consider this equilibrium fair -- perhaps receiving market rates is undesirable as well -- but at least solely in the context of this game, there's no reason to call it unfair.  The classifier is the same but the normative implications are wildly different, implying no definition of fairness can distinguish between these two scenarios using only the classifier.  This is a significant issue because the existing computational definitions of fairness, discussed in Section \ref{sec:background}, are defined as a function of the classifier.

It might be possible to attempt to adopt computational or economic definitions to this setting in an attempt to circumvent this outcome.  It might be possible to label the classifier as unfair in both equilibria using some version of an existing definition, for example.  However, in Section \ref{sec:fair_failures}, we will provide evidence that existing fairness definitions will label classifiers employed in unfair equilibria as fair, ruling out existing definitions and making small tweaks to existing definitions unlikely to be effective.  Moreover, these definitions will do this not just for this particular game, but for a wide class of games.  To do so, we need to introduce a formal definition of \emph{blatant unfairness} for equilibria.  We do not attempt to define fairness here, simply blatant unfairness, similarly to the way that Dinur and Nissim \cite{DinurN03} did not define privacy, simply blatant non-privacy.

This definition, for any equilibrium in a two-player complete-information game, says that the equilibrium is blatantly unfair when it has the same features as scenario A did, which we had already decided was blatantly relationally unfair:  in scenario A, there was a player $x$ who did not get positive utility at equilibrium, and yet there was another equilibrium where both players received positive utility.  We do not insist on $x$ having any particular features, just as in scenario A the outcome did not depend on the surplus generated or their outside options.

\begin{defn}\label{def:2p_blatant_unfairness}
An equilibrium $\pi$ is \emph{blatantly unfair} with respect to $x$ in a two-player complete-information game with players $x$ and $y$ if the payoff for $x$ under $\pi$ is non-positive, but there is an equilibrium where both $x$ and $y$ receive positive payoffs.
\end{defn}

The equilibrium $\pi$ where $x$ receives non-positive payoff is scenario A, whereas the alternative equilibrium where both players receive positive playoffs is scenario B.  If this definition is to be successful, it must not only capture the blatantly relationally unfair equilibria we already described, but at least some of the larger aims of Andersonian relational equality, and here, we focus on capturing one form of unequal standing in a social relationship.  As such, the game in this definition represents the social relationship between players $x$ and $y$:  the outcomes of the game determine how each player's actions and beliefs affect the other player.  Under $\pi$, the player $y$ exerts sufficient pressure to ensure that any response by player $x$ results in a negative utility for $x$.  This pressure is modeled by the complete-information game itself, and since the game is arbitrary, this is sufficiently general to capture both some direct social pressure, like a threat, or something more indirect, like they both believe $x$ suffers from racial discrimination.  Either way, the pressure results from the relationship itself, as brought about by the actions they play, rather than any outside forces, as there is an alternative equilibrium where they both cooperate with each other and both receive positive payoffs.

We refer to this equilibrium as \emph{blatantly} relationally unfair because the payoff for $x$ is not just suboptimal, it is negative, and both players know it.  We use the term \emph{blatant unfairness}, rather than blatant \emph{relational} unfairness, to highlight that this definition captures forms of distributional inequality as well (it is a function of payoffs).  When referring to a classifier $g$ played at a blatantly unfair equilibrium, to indicate that the equilibrium is blatantly unfair when that classifier is employed as a strategy in at least one game, we will refer to the classifier itself as \emph{blatantly unfair}, simply for convenience.

Anderson's critique of distributional equality motivates the features of this definition.  We insist on an absolute level of welfare -- utility must be positive if possible -- in order to prevent neglect of the social equality of workers.  We insist on involving multiple players to capture relationships rather than just outputs.  We insist on examining entire equilibria, rather than classifiers, to prevent decision makers from claiming fairness in narrow settings while manipulating who gets to opt in and who get to opt out of those settings.  And to avoid moral paternalism, we insist that establishing blatant unfairness does not require justifying the moral deserts of the candidates (or firms) by first establishing that they are sufficiently qualified, or skillful, or possess any particular inherent properties or attributes.  This last insistence also has the beneficial effect that our definition, unlike group fairness, does not necessarily require gathering data about sensitive attributes, or force us to define what is or should be a sensitive attribute.  Nonetheless, our definition of blatant unfairness should help protect minoritized populations from discriminatory outcomes, including racial discrimination.  At the very least, since racial discrimination may often be encoded as beliefs about outside options, preventing blatant unfairness is necessary for preventing racial discrimination in hiring markets.  It may not be sufficient for preventing racial discrimination, depending on how the outcomes of the market affect other relationships outside of the market.  Our focus in this work, though, is expanding the domain where discrimination is evaluated from the outcome of a single classifier to a much larger system: the market. 

Our approach also attempts to encode the fact that ``democratic equality regards two people as equal when each accepts the obligation to justify their actions by principles acceptable to the other, and in which they take mutual consultation, reciprocation, and recognition for granted.'' \cite{Anderson:1999wh}  If two people accept this obligation, they must actually find a principle by which actions may really be acceptable to the other, and thus we attempt to define such a principle, at least under this particular set of simplifying assumptions.  Under $\pi$, player $y$ has failed to take actions which could ever be acceptable to $x$.

Since we care not just about the relative utility that the two players get, but whether they receive negative payoffs, the utility functions are absolute, in the sense that they cannot be arbitrarily renormalized up to ordering of preferences over actions.  Rather, to define blatant unfairness, our assumption is that there is a space of payoffs that are acceptable to a player, and that 1) all payoffs acceptable to a player are larger than all payoffs that are unacceptable to that player and 2) both players agree on what each other's "acceptable space" of payoffs is.  These assumptions prevent us from using a utility function that is everywhere non-negative, such as declaring the payoff for the candidate to be the salary of the job, which is always non-negative, without considering that they start with a need for a job, and therefore negative utility.  Any utility function with an acceptable space of utilities can then be renormalized to set the threshold between acceptable and non-acceptable utilities to be zero, which is why we use this threshold for the definition of blatant unfairness.  Some additional assumptions along these lines are necessary anyway in this setting, as Andersonian relational equality requires finding a principle of action that is acceptable to the other, thus at the very least necessitating each player to know what is acceptable to themselves.  Complete-information games are thus a very natural starting place, because they already agree on each other's payoffs.

Implicitly, this definition is parameterized not only by the threshold at which payoffs are no longer acceptable, but also the kind of equilibrium, here SCEs.  This places significant normative weight on which kind of equilibrium we use as a solution concept.  The choice of which kind of equilibrium to use -- should they be approximate equilibria or exact? -- could make big differences to which strategy profiles are blatantly unfair.  And small changes to a game could erase or create equilibria, which would change whether a given strategy profile is blatantly unfair or not.  This appears to be an unavoidable choice to have to make:  the status of a relationship depends not only on a person's behavior now, but how another person reacts to changes in that behavior.  They suffer from low standing in a relationship not necessarily because the other person is causing them harm, but because that other person holds attitudes about them that makes their inability to renegotiate the relationship a consistent feature of that relationship.  Here, SCEs were chosen to better represent the effect of beliefs on action than Nash equilibria or the like, which explains how payoffs can fail to reflect the underlying qualifications of a candidate, but certainly other choices should be investigated in the future.

Perhaps counterintuitively, our use of SCEs means it is sometimes blatantly unfair for the firm to have ``correct'' beliefs:  if the firm correctly believes the candidate's outside option is low, then the Nash equilibrium where they get market rates results in the candidate getting negative utility.  But there is an equilibrium where both firm and candidate both receive positive utility, one where the firm deludes itself into believing the candidate's outside option is high, gives it to them, which ensures the firm never observes the market.  But just because the former equilibrium was unfair, doesn't mean we should move to the latter equilibrium: labeling an equilibrium blatantly unfair does not make the equilibrium where both get positive payoffs fair.  In this example, it might instead mean changing the underlying market conditions.

\subsection{The failures of group fairness and other definitions}\label{sec:fair_failures}

In this section, we show that it is possible for a classifier to be part of a blatantly unfair equilibrium, but simultaneously satisfy group fairness and other well-known definitions of fairness.  When we interpret these other definitions as moral axioms, this is an incompatibility result, showing that satisfying any one of these definitions of fairness does not imply non-blatant unfairness.  But if we are convinced by the examples detailed above, then this result is stronger:  The examples enumerated in this result show why group fairness and other definitions fail to prevent unfairness.

To do so, we consider a game $G_{x,f}$ like before, but now the structure of the game is arbitrary, except that $f$ still plays a classifier $g$ as their strategy, each $g(x)$ their strategy for $G_{x,f}$.  We demonstrate that as long as a blatantly unfair equilibrium $\tilde{\pi}$ exists in this game, we can choose some player $x_0$ to receive that strategy, i.e. $g(x_0) = \tilde{\pi}_f$, so that $g$ is blatantly unfair, and then extend $g$ to the rest of the players so that it is group fair, or individually fair, etc.  Technically, this is not a particularly difficult result: there is considerable leeway to extend $g$, since the classifier is not constrained in any way on the remaining players $x\in X$.  Indeed, we may extend this result to the setting where each player $x$ can play in a game with different structure and payoffs, but for notational simplicity we do not do so here.

\begin{prop}\label{prop:incompatibility}
Consider a two-player complete-information game $G_{x,f}$ with players $x\in X$ and $f$, and $f$'s pure strategies are of the form $g(x)$, where $g:X\rightarrow D$, $D\subset\mathbb{R}$.  Suppose there exists a blatantly unfair equilibrium $\tilde{\pi}$ in this game.  Then:
\begin{enumerate}
\item If $D$ is finite or a real interval, and for all $y$ and for all $z$ in the image $F_1(D,Y(X))$ there exists a prediction $d$ such that $F_1(d,y) = z$, then there is a strategy $g$ for $f$ that is $(F_1,F_2)$-group fair and blatantly unfair.
\item If $\tilde{\pi}_f$ is a pure strategy, $|Y(X)|\le |F_2(D,Y(X))|$, and for all $y$ and for all $z$ in the image $F_2(D,Y(X))$ there exists a prediction $d$ such that $F_2(d,y) = z$, then then there is a strategy $g$ for $f$ that is $(F_1,F_2)$-group fair and blatantly unfair.
\item For all $M,m$, there is a strategy $g$ for $f$ that is $(M,m)$-individually fair and blatantly unfair.
\item Consider a causal graph which includes as nodes all of $X$ as well as a sensitive feature $A$ and the output of the strategy, i.e.\ $D$.  Suppose $A$ has out-degree $0$.  Then there is a strategy $g$ for $f$ which, with respect to $A$, is individually causally fair, exhibits no taste-based discrimination, and is blatantly unfair.
\end{enumerate}
\end{prop}

In addition to an incompatibility result, the strategies $\pi$ that evince this incompatibility are also examples that we believe demonstrate some of the failures of these fairness definitions.  Since blatant unfairness is an individual-level definition, it is hopefully intuitive why group fairness should not suffice to prevent blatant unfairness.  But there is another fundamental issue with group fairness.  For example, consider statistical parity, which requires $P[g(x) \le z | A =a] = P[g(x) \le z]$ for all $a\in A$ and $z\in D$.  Part (1) implies that for \emph{any} sensitive attribute $A$ and distribution over $X$, there is a blatantly unfair strategy $g$ that satisfies statistical parity.  We set $g(x)$ to be the same for all $x$, namely the (possibly mixed) strategy $\tilde{\pi}_f$, the strategy for $f$ guaranteed to be blatantly unfair by assumption.  Constant functions always satisfy statistical parity, which only requires the same average treatment to all groups defined by $A$.

While we were driven to consider relational equality in part because group fairness risks paternalistic control and ignoring social inequality, here we see that statistical parity is also relationally unfair despite guaranteeing equality of outcome when everybody is treated equally badly.  This also holds for the rest of group fairness and EOP, since these are also content to treat fairness (or justice, or equality) as a comparison between people. For example, while equalized odds ($F_1(d,Y(x)) = d$ and $F_2(d,Y(x))=Y(x)$) was introduced to correct perceived failures of statistical parity \cite{Hardt:2016wv}, our concerns hold just as true for equalized odds.  Even if we ignore existing concerns about equalized odds, it is still a group-level constraint that permits everybody, including those qualified, to be treated badly, so long as it does so on average equally across $A$, given those qualifications.

Part (1) of Proposition \ref{prop:incompatibility} covers a generalization to any group fairness where $F_1$ is surjective via its first coordinate, which includes statistical parity, equalized odds, accuracy equality, etc.  To show this, we note that it suffices to use a strategy which is not necessarily a constant, but on every $x$, the distribution $F_1(g(x),Y(x))$ is the same.  

\begin{proof}[Proof of (1)]
First, suppose for all $y\in Y$ and for all $z$ in the image $F_1(D,Y(X))$ there exists a prediction $d\in D$ such that $F_1(d,y) = z$.  Consider the function $F_1(\cdot,y)$ onto its image.  Since it is by definition surjective, it always has a right inverse, i.e.\ a function $i_y:F_1(D,y)\rightarrow D$ such that for any $z\in F_1(D,y)$, $F_1(i_y(z),y) = z$.  Moreover, a measurable right inverse exists when $D$ is finite or a real interval -- more generally, it can be any uncountable Polish space \cite{lowther_2016} -- so we may assume $i_y$ is measurable.  Let $x_0\in X$.  Then let $g(x_0)= \tilde{\pi}_f$ and for all other $x\in X$, let $g(x) = i_{Y(x)}(F_1(\tilde{\pi}_f,Y(x_0)))$.  By supposition, $g$ is blatantly unfair.  This is well defined when the support of $F_1(\tilde{\pi}_f,Y(x_0))$, which is contained in $F_1(D,Y(x_0))$, is in the domain of $i_{Y(x)}$.  So it suffices to be able to find, for any $z\in F_1(D,Y(x_0))$, a $d\in D$ such that $F_1(d,Y(x))=z$, which we have already assumed.

Thus we can conclude $F_1(g(x),Y(x)) = F_1(\tilde{\pi}_f,Y(x_0))$ for all $x$, and so for any $a\in A$,
\begin{align*}
& P[F_1(g(x),Y(x)) \le z | A(x)=a, F_2(g(x),Y(x))] \\
&= P[F_1(\tilde{\pi}_f,Y(x_0)) \le z | A(x)=a, F_2(g(x),Y(x))]\\
&= P[F_1(\tilde{\pi}_f,Y(x_0)) \le z],
\end{align*}
where the probabilities are over the randomness of $x$ and the randomness of the functions, and the last equality follows from the independence of $F_1(\tilde{\pi}_f,Y(x_0))$ from $x$.  Hence $g$ is $(F_1,F_2)$-group fair.
\end{proof}

While part (1) covers a wide range of group fairness definitions, it does not include all popular group fairness definitions, most notably \emph{sufficiency} \cite{barocas-hardt-narayanan}, also known as fair \emph{calibration} \cite{Chouldechova17}, where $F_1(d,Y(x)) = Y(x)$ and $F_2(d,Y(x)) = d$.  We cannot use a constant function anymore, as it will not satisfy sufficiency except under very narrow conditions on $Y$.  The easiest way to ensure sufficiency is if we instead assign a unique decision $d_y$ for every possible label $y\in Y(X)$ (and in general, all we need is a distinct value of $F_2$ for every label, which requires $\tilde{\pi}_f$ to be pure to ensure uniqueness).  To some $x_0$, we assign the decision $g(x_0) = \tilde{\pi}_f$, guaranteeing $g$ is blatantly unfair, and ensure that no other player $x$ with a different label from $x_0$ receives that same decision.
\begin{proof}[Proof of (2)]
Now, suppose for all $y$ and for all $z$ in the image $F_2(D,Y(X))$ there exists a prediction $d$ such that $F_2(d,y) = z$.  If $|Y(X)|\le |F_2(D,Y(X))|$, then there exists an injection $i:Y(X)\rightarrow F_2(D,Y(X))$.  For $y\in Y(X)$, let $d_y$ be the prediction satisfying $F_2(d_y,y) = i(y)$.  Let $x_0\in X$.  Now we construct $g$, and there are three cases depending on how $\tilde{\pi}_f$ coincides with the decisions $d_y$.  If it happens to be the case that $d_{Y(x_0)} = \tilde{\pi}_f$, let $g(x) = d_{Y(x)}$ for all $x\in X$.  If there is no $x\in X$ for which $d_{Y(x)} = \tilde{\pi}_f$, let
\[g(x)=
\begin{cases}
d_{Y(x)} \text{ if } Y(x)\neq Y(x_0)\\
\tilde{\pi}_f \text{ if } Y(x)=Y(x_0).
\end{cases}\]  Otherwise, call $x_1$ the value for which $d_{Y(x_1)} = \tilde{\pi}_f$.  Then let
\[g(x)=
\begin{cases}
d_{Y(x)} \text{ if } Y(x)\neq Y(x_0), Y(x_1)\\
d_{Y(x_1)} \text{ if } Y(x) = Y(x_0)\\
d_{Y(x_0)} \text{ if } Y(x) = Y(x_1).
\end{cases}\]
By construction, $g(x_0) = \tilde{\pi}_f$ and $g$ is blatantly unfair.  Let $x_j,x_k \in X$.  if \[F_2(g(x_j),Y(x_j)) = F_2(g(x_k),Y(x_k)),\] since $i$ was injective, $Y(x_j)=Y(x_k)$.  By construction of $g$ this in turn implies $g(x_j)=g(x_k)$ and so $F_1(g(x_j),Y(x_j)) = F_1(g(x_k),Y(x_k))$.  Thus for any $a\in A$,
\begin{align*}
& P[F_1(g(x),Y(x)) \le z | A(x), F_2(g(x),Y(x))] \\
= & P[F_1(g(x),Y(x)) \le z | F_2(g(x),Y(x))]
\end{align*} 
and $g$ is $(F_1,F_2)$-group fair.

\end{proof}

Between parts (1) and (2), this result covers a very wide range of group fairness definitions.  This result is not true for all group fairness definitions, as there are values of $F_1$ and $F_2$ for which no classifiers exist that satisfy $(F_1,F_2)$-group fairness except under narrow conditions (e.g.\ $F_1(d,Y(x)) = Y(x)$ and $F_2(d,Y(x))$ is a constant), or for which no classifiers exist that satisfy $(F_1,F_2)$-group fairness after fixing a value $g(x_0)$ (e.g.\ $F_1(d,Y(x)) = d+Y(x)$ and $F_2(d,Y(x))$ is a constant, but $d,Y\in\{0,1\}$).  So we have instead given simple, natural conditions that cover a wide range of popular group fairness definitions.

Both Dworkian individual fairness and individual counterfactual fairness fall victim to the same issue that group fairness did:  treating every individual equally is still unfair when you treat every individual equally badly.  Constant functions are always individually fair.  And the output of a constant function does not change when you change the value of the sensitive attribute, so it is also individually counterfactually fair.

\begin{proof}[Proof of (3)]
It suffices to consider the constant function that fixes $\tilde{\pi}_f$, i.e.\ let $g(x) = \tilde{\pi}_f$ for all $x\in X$.  Then $M(g(x),g(y)) = 0 \le m(x,y)$ for all $x,y\in X$.
\end{proof}

\begin{proof}[Proof of (4)]
Again it suffices to consider the constant function that fixes $\tilde{\pi}_f$, i.e.\ let $g(x) = \tilde{\pi}_f$ for all $x\in X$.  $g$ only depends on $\tilde{\pi}_f$, which in turn cannot depend on $A$, because it has out-degree 0.  No path in the causal graph between $A$ and the output of $g$ immediately implies that it is both individually causally fair and exhibits no taste-based discrimination.
\end{proof}

Taste-based and statistical discrimination take a fundamentally different approach to defining (un)fairness from definitions predicated on prediction error. We might hope that by eliminating both taste-based and statistical discrimination, we could prevent blatantly unfair equilibria -- we would thus be able to \emph{define} discrimination as the combination of taste-based and statistical discrimination.  Unfortunately, Proposition \ref{prop:sce} already implies that this is not the case.  The constant classifier $g(x) = 0$ is an SCE in $M_{x,f}$ (it may be played at equilibrium for all $x$) which results in negative utility for all $x$, and there is an alternative classifier $g(x) = 3/2$ which is an SCE that results in positive utility for all $x$, so $g(x) = 0$ is a blatantly unfair equilibrium strategy.  This blatantly unfair classifier is constant, so is not causally dependent on any sensitive attribute, and thus is not an example of taste-based discrimination.  Moreover, all $x$ create the same surplus with the firm, and so this classifier cannot be an example of statistical discrimination.  This is a concrete example where taste-based and statistical discrimination represent a very different kind of inequality than relational inequality.

Since $M_{x,f}$ has a blatantly unfair equilibrium $\tilde{\pi}$ where $\tilde{\pi}_f = 0$, this classifier satisfies the fairness definitions in parts (1), (3), and (4) of Proposition \ref{prop:incompatibility}.  Assuming that the labels $Y(x)$ are a constant because each candidate's surplus is identical, the fairness definitions in part (2) are also satisfied by the constant classifier $g(x) = \tilde{\pi}_f = 0$.

\begin{cor}
For any sensitive feature $A$ and constant labels $Y(x)$, there is a blatantly unfair equilibrium in $M_{x,f}$ where $f$ uses a constant classifier at equilibrium, but this classifier simultaneously satisfies statistical parity, equalized odds, sufficiency, individual fairness, and displays no taste-based or statistical discrimination.
\end{cor}

If we want to prevent blatantly unfair equilibrium, we cannot use any of these definitions of fairness.  Yet if the only reason to introduce this conception of unfairness in games is to point out that guaranteeing EOP is unhelpful when no candidate receives sufficient resources, then existing remedies would be sufficient.  In particular, we could take a welfare-based approach, where we try to ensure that the average (or minimum, or some function) of the utilities the players receive is as large as possible.  This kind of utilitarianism has been expounded on in the context of algorithmic fairness by Heidari et al.\ \cite{HeidariFGK18}.  If we for example insist on maximizing on welfare, and define welfare as the minimum share of the surplus, in our example, this would force an equal split of the surplus, and this outcome would not be blatantly unfair.  Indeed, blatant unfairness as defined is certainly compatible with a welfare-based approach:  we are insisting that welfare be past a certain threshold.

Nonetheless, we have a substantially different motivation than welfare-based approaches do.   Since approaches based on welfare are typically grounded in EOP by defining welfare as some function of the utilities of individuals (as in Heidari et al.\ \cite{HeidariFGK18}), these approaches, just like group, individual, and causal fairness definitions, fall victim to Anderson's criticisms of EOP, including neglecting the social equality of workers and enacting a paternalistic oversight by those using such approaches.

Where we should see welfare-based approaches and relational approaches truly diverge, though, is in larger, more complex systems than complete-information games.  In complete-information games, the firm and candidate have already implicitly agreed that both of their goals are to maximize the utilities and they agree empirically on how each possible set of actions will produce different utilities.  Each already recognizes the other's needs, as represented by these utilities, even if they are not a priori willing to satisfy those needs.  
Because each agrees to the other's utilities, the only way to justify their actions is according to these utilities:  hence an equilibrium is relationally unfair when they are not incentivized to meet each other's needs.  Andersonian relational equality, though, is a political theory best suited to more complex domains without complete information, where players do not necessarily agree on each other's payoffs.  Players can communicate with each other and have uncertainty about each other and the future.  Each player may have various sorts of uncertainty about their own goals -- they may have multiple competing goals or may not even know what their own goals should be.  This is the world in which relational equality is intended for, and we pose as an open question. 

\section{Beyond two-player games}\label{sec:multi-player}

Implicitly, group fairness appears designed for situations with limited resources, at least when the labels represent decisions about resources.  If there were no limits to resources, we could employ a classifier that gave out resources to everybody and there would be no need to equalize the distribution of the resources across groups.  The firm in Section \ref{sec:blatant_unfairness}, however, is not resource-constrained because they produce unit surplus with every candidate.  Group fairness appears entirely unnecessary anyway in this case.  So it is worth investigating what happens in the resource-constrained case, where firms can only produce surplus with a limited number of candidates.  This kind of constraint is best modeled in a game with many players, which motivates the need to expand the definition of blatant unfairness beyond two-player games.

We extend blatant unfairness by considering our previous example.  We could have considered the set of games $\{M_{x,f} : x\in X\}$, previously examined in Section \ref{sec:blatant_unfairness} in isolation, as a single game.  Whereas before, we posited no particular relationship between the games, here we can consider, for example, what happens if the games happen simultaneously:  The firm plays a strategy $g:X\rightarrow[0,3]$, and then every candidate $x$ simultaneously decides whether or not to accept the offer $g(x)$.  If they accept, as before, they receive $g(x)-1$ utility, and the firm receives $2-g(x)$ utility, i.e.\ the firm splits unit surplus with every candidate.  If they reject, as before, they receive an exogenous but unknown outside option.  Call this game $M_X$.

Since each candidate plays an independent subgame, the equilibrium strategies we found in Proposition \ref{prop:sce} exist here in each subgame.  Thus, if we believe that any such individual subgame with such an equilibrium is blatantly unfair, as we have already asserted, then we should believe the entire game is blatantly unfair regardless of the outcomes of all the other subgames.  This motivates the following definition, which says that all players should be like the firm in this example, which faces a constant, positive-sum subgame with the candidate for whom the equilibrium is blatantly unfair, or else be like the other candidates, who are unaffected:

\begin{defn}\label{def:multi_blatant_unfairness}
An equilibrium $\pi$ is \emph{blatantly unfair} with respect to $x$ in a complete-information game if the payoff for $x$ under $\pi$ is not positive, but there is an equilibrium $\pi'$ where $x$ receives a positive payoff, and for any other player, they either receive a payoff at least as high under $\pi'$ as under $\pi$, or receive a positive payoff under $\pi'$.
\end{defn}

Alternatively, we could define a welfare function as the following:  one set of payoffs is as least as preferred as another if every player with negative utility receives at least as much utility, and every player with positive utility still receives positive utility.

$M_X$, introduced above, witnesses the existence of blatantly unfair equilibria in multi-player games where the firm plays a group fair classifier (or an individually fair classifier, etc.), though we leave this to the reader to verify.  Proposition \ref{prop:incompatibility} cannot carry over to this setting, however:  there might not be any equilibria that are group-fair, let alone the particular equilibria that happens to be blatantly unfair.

This definition is not the only possibility for extending Definition \ref{def:2p_blatant_unfairness}.  We could consider variants, such as requiring that payoffs for all players be positive in the alternative equilibrium, for example.  We do not focus on such variations because we believe that distinctions should be made by examining non-complete-information games and other more realistic settings, rather than attempting to catalog all possible variations as has been done for group fairness.  We leave this for future work.

In any case, consider modifying $M_X$ by introducing a constraint on resources via a limit on the number of jobs the firm can offer.  This may be forced due to diminishing marginal returns to hiring additional people, or due to the capital overhead hiring requires.  For example, suppose the firm only offers an equal split of the surplus, or none at all, i.e. $g(x)\in\{0,3/2\}$, and can only offer some $i<|X|$ people the job, i.e. only $i$ people may receive a non-zero offer.

If the firm is a monopoly, so all outside options are zero, then no $g(x)$ that offers $i$ people the job is part of a blatantly unfair equilibrium:  anyone currently not being offered a job can only be offered a job at the expense of someone else's offer, who has no outside option because the firm is a monopoly.  This makes resource-constrained monopolies the ``natural'' setting for group fairness (and individual fairness, etc.).  But even in this case, this may not be a strong argument for the use of group fairness in such cases, but rather an argument for a normative requirement to modify the setting via policy-making when it appears.

As soon as we jettison the assumption that the firm is a monopoly, either by adding non-zero exogenous outside options, or adding more firms so that the market can clear, it's again possible for blatantly unfair equilibria to exist.  For example, in the simplest case, if there's some $|X|$ firms, even if each firm can only hire one candidate, where outside options are as before (representing either exogenous or endogenous outside options), then this is just equivalent to our original model in Section \ref{sec:blatant_unfairness} and while blatant unfairness can exist in this market, that also means there are equilibria where everyone receives surplus, unlike the monopoly case.  While the constraints firms face may appear at first blush to motivate previous definitions of group fairness, those definitions still fail to capture blatant unfairness in markets. 

\section{Limitations and Future Work}\label{sec:discussion}

In this work, we have shown how the existence of self-confirming equilibria in hiring markets create blatantly unfair situations in which candidates are trapped in an undesirable relationship with a firm.  This scenario enabled us to define a formal model of blatant unfairness as equilibria of this kind, and show how existing technical definitions of fairness fail to escape these equilibria. However, this work is not without limitations, a few of which we discuss now.

While we have set out to develop a formal model of relational equality, we have not completed this goal.  Indeed, as we discuss in Section \ref{sec:fair_failures}, our approach is largely consistent with welfare-based approaches, though its motivation is distinctly different.  The difference is, unlike in traditional social choice theory, we focus solely on sets of payoffs that can happen at equilibrium.  We do not expect such a close relationship to social choice theory outside of games where players agree on each other's utilities, but we leave this problem for future work.

Critically for our purposes, Anderson has noted the necessary interrelationship between distributional and relational forms of equality: for Anderson, distributional or material equality is a potentially necessary but insufficient condition for relational equality to occur.  As such, Andersonian relational equality is concerned with distributional equality to the extent it ``requires that everyone have effective access to enough resources to avoid being oppressed by others and to function as an equal in civil society'' \cite{Anderson:1999wh} (320). Such an analysis of the relationship between social power and material inequality would require looking at dynamics, which we leave for future work. Anderson notes that some degree of material equality must be necessary in order to ensure that elites, whether of wealth, status, or other manifestations of social power, do not gain outsized influence in a particular society, and we demonstrate a very simple version of this principle:  relationally unequal situations in games with public knowledge of payoffs can be avoided by ensuring payoffs are positive at equilibrium.  More sophisticated versions of this principle will require examining systems with a wider set of actions, including being able to exert control over the rules of the game.  Blatant unfairness is defined with respect to the actions players can take, and the most impactful actions are those that change the game entirely.

\bibliographystyle{plain}
\bibliography{refs}

\begin{thebibliography}{10}

\bibitem{AbbasiFSV19}
Mohsen Abbasi, Sorelle~A. Friedler, Carlos Scheidegger, and Suresh
  Venkatasubramanian.
\newblock Fairness in representation: quantifying stereotyping as a
  representational harm.
\newblock In Tanya~Y. Berger{-}Wolf and Nitesh~V. Chawla, editors, {\em
  Proceedings of the 2019 {SIAM} International Conference on Data Mining, {SDM}
  2019, Calgary, Alberta, Canada, May 2-4, 2019}, pages 801--809. {SIAM}, 2019.

\bibitem{Agre:1997ts}
Philip~E. Agre.
\newblock {Toward a Critical Technical Practice: Lessons Learned in Trying to
  Reform AI}.
\newblock In Geoffrey~C. Bowker, Les Gasser, Susan~Leigh Star, and Bill Turner,
  editors, {\em Social Science, Technical Systems and Cooperative Work: The
  Great Divide}, pages 131--158. Hillsdale, NJ, 1997.

\bibitem{Anderson:2007uo}
Elizabeth Anderson.
\newblock {Fair Opportunity in Education: A Democratic Equality Perspective}.
\newblock {\em Ethics}, 117:595--622, July 2007.

\bibitem{Anderson:2008gw}
Elizabeth Anderson.
\newblock {I{\textemdash}Elizabeth Anderson: Expanding the Egalitarian Toolbox:
  Equality and Bureaucracy}.
\newblock {\em Aristotelian Society Supplementary Volume}, 82(1):139--160, June
  2008.

\bibitem{Anderson:2013hu}
Elizabeth Anderson.
\newblock {The Fundamental Disagreement between Luck Egalitarians and
  Relational Egalitarians}.
\newblock {\em Canadian Journal of Philosophy}, 40(sup1):1--23, July 2013.

\bibitem{Anderson:2015wq}
Elizabeth Anderson.
\newblock {Liberty, Equality, and Private Government}.
\newblock The Tanner Lectures in Human Values, Princeton University, August
  2015.

\bibitem{Anderson:1999wh}
Elizabeth~S Anderson.
\newblock {What Is the Point of Equality?}
\newblock {\em Ethics}, 109(2):287--337, January 1999.

\bibitem{Arrow:1971tv}
Kenneth Arrow.
\newblock {The Theory of Discrimination}.
\newblock In {\em Discrimination in Labor Markets}, pages 1--37. Princeton
  University, October 1971.

\bibitem{arrow1998has}
Kenneth~J Arrow.
\newblock What has economics to say about racial discrimination?
\newblock {\em Journal of economic perspectives}, 12(2):91--100, 1998.

\bibitem{Barabas:2018vc}
Chelsea Barabas, Madars Virza, Karthik Dinakar, Joichi Ito, and Jonathan
  Zittrain.
\newblock {Interventions over Predictions - Reframing the Ethical Debate for
  Actuarial Risk Assessment.}
\newblock In {\em FAT* 2018}, 2018.

\bibitem{barocas-hardt-narayanan}
Solon Barocas, Moritz Hardt, and Arvind Narayanan.
\newblock {\em Fairness and Machine Learning}.
\newblock fairmlbook.org, 2019.
\newblock \url{http://www.fairmlbook.org}.

\bibitem{Benthall:2019dpa}
Sebastian Benthall and Bruce~D Haynes.
\newblock {Racial categories in machine learning}.
\newblock In {\em Conference on Fairness, Accountability, and Transparency
  2019}, pages 289--298, New York, New York, USA, 2019. ACM Press.

\bibitem{bertrand2004emily}
Marianne Bertrand and Sendhil Mullainathan.
\newblock Are emily and greg more employable than lakisha and jamal? a field
  experiment on labor market discrimination.
\newblock {\em American economic review}, 94(4):991--1013, 2004.

\bibitem{Binns:2018tf}
Rueben Binns.
\newblock {Fairness in Machine Learning: Lessons from Political Philosophy}.
\newblock {\em Proceedings of Machine Learning Research}, 81:1--11, 2018.

\bibitem{Birhane21}
Abeba Birhane.
\newblock Algorithmic injustice: a relational ethics approach.
\newblock {\em Patterns}, 2(2):100205, 2021.

\bibitem{bogen2018help}
Miranda Bogen and Aaron Rieke.
\newblock Help wanted: An examination of hiring algorithms, equity, and bias.
\newblock 2018.

\bibitem{Chouldechova17}
Alexandra Chouldechova.
\newblock Fair prediction with disparate impact: {A} study of bias in
  recidivism prediction instruments.
\newblock {\em Big Data}, 5(2):153--163, 2017.

\bibitem{Cohen:1989td}
G.~A. Cohen.
\newblock {On the Currency of Egalitarian Justice}.
\newblock {\em Ethics}, 99(4):906--944, July 1989.

\bibitem{CorbettDavies:2018ud}
Sam Corbett-Davies and Sharad Goel.
\newblock {The Measure and Mismeasure of Fairness: A Critical Review of Fair
  Machine Learning}.
\newblock pages 1--25, September 2018.

\bibitem{CostanzaChock:2018fc}
Sasha Costanza-Chock.
\newblock {Design Justice: towards an intersectional feminist framework for
  design theory and practice}.
\newblock In {\em Design Research Society}, pages 1--14, 2018.

\bibitem{cowgill2019economics}
Bo~Cowgill and Catherine~E Tucker.
\newblock Economics, fairness and algorithmic bias.
\newblock {\em preparation for: Journal of Economic Perspectives}, 2019.

\bibitem{DAmourSABSH20}
Alexander D'Amour, Hansa Srinivasan, James Atwood, Pallavi Baljekar,
  D.~Sculley, and Yoni Halpern.
\newblock Fairness is not static: deeper understanding of long term fairness
  via simulation studies.
\newblock In Mireille Hildebrandt, Carlos Castillo, L.~Elisa Celis, Salvatore
  Ruggieri, Linnet Taylor, and Gabriela Zanfir{-}Fortuna, editors, {\em FAT*
  '20: Conference on Fairness, Accountability, and Transparency, Barcelona,
  Spain, January 27-30, 2020}, pages 525--534. {ACM}, 2020.

\bibitem{Dencik:2019jj}
Lina Dencik, Arne Hintz, Joanna Redden, and Emiliano Trer{\'e}.
\newblock {Exploring Data Justice: Conceptions, Applications and Directions}.
\newblock {\em Information, Communication {\&} Society}, 22(7):873--881, May
  2019.

\bibitem{Denton_Hanna_Amironesei_Smart_Nicole_2021}
Emily Denton, Alex Hanna, Razvan Amironesei, Andrew Smart, and Hilary Nicole.
\newblock On the genealogy of machine learning datasets: A critical history of
  imagenet.
\newblock {\em Big Data and Society}, 2021.

\bibitem{DinurN03}
Irit Dinur and Kobbi Nissim.
\newblock Revealing information while preserving privacy.
\newblock In Frank Neven, Catriel Beeri, and Tova Milo, editors, {\em
  Proceedings of the Twenty-Second {ACM} {SIGACT-SIGMOD-SIGART} Symposium on
  Principles of Database Systems, June 9-12, 2003, San Diego, CA, {USA}}, pages
  202--210. {ACM}, 2003.

\bibitem{Doucet_2015}
Andrea Doucet.
\newblock Parental responsibilities: Dilemmas of measurement and gender
  equality.
\newblock {\em Journal of Marriage and Family}, 77(1):224?242, 2015.

\bibitem{du2020fairness}
Mengnan Du, Fan Yang, Na~Zou, and Xia Hu.
\newblock Fairness in deep learning: A computational perspective.
\newblock {\em IEEE Intelligent Systems}, 2020.

\bibitem{Dwork:2011vl}
Cynthia Dwork, Moritz Hardt, Toniann Pitassi, Omer Reingold, and Richard Zemel.
\newblock {Fairness Through Awareness}.
\newblock pages 1--24, November 2011.

\bibitem{DworkR14}
Cynthia Dwork and Aaron Roth.
\newblock The algorithmic foundations of differential privacy.
\newblock {\em Found. Trends Theor. Comput. Sci.}, 9(3-4):211--407, 2014.

\bibitem{Dworkin:1981vu}
Ronald Dworkin.
\newblock {What is Equality? Part 1: Equality of Welfare}.
\newblock {\em Philosophy {\&} Public Affairs}, 10(3):185--246, 1981.

\bibitem{Dworkin:1981ux}
Ronald Dworkin.
\newblock {What is Equality? Part 2: Equality of Resources}.
\newblock {\em Philosophy {\&} Public Affairs}, 10(4):283--345, 1981.

\bibitem{ElzaynJJKNRS19}
Hadi Elzayn, Shahin Jabbari, Christopher Jung, Michael~J. Kearns, Seth Neel,
  Aaron Roth, and Zachary Schutzman.
\newblock Fair algorithms for learning in allocation problems.
\newblock In danah boyd and Jamie~H. Morgenstern, editors, {\em Proceedings of
  the Conference on Fairness, Accountability, and Transparency, FAT* 2019,
  Atlanta, GA, USA, January 29-31, 2019}, pages 170--179. {ACM}, 2019.

\bibitem{Eubanks:2018wv}
Virgina Eubanks.
\newblock {\em {Automating Inequality: How High-Tech Tools Profile, Police, and
  Punish the Poor}}.
\newblock St. Martin's Press, New York, 2018.

\bibitem{fang2011theories}
Hanming Fang and Andrea Moro.
\newblock Theories of statistical discrimination and affirmative action: A
  survey.
\newblock {\em Handbook of social economics}, 1:133--200, 2011.

\bibitem{Fish:2019gi}
Benjamin Fish, Ashkan Bashardoust, danah boyd, Sorelle~A. Friedler, Carlos
  Scheidegger, and Suresh Venkatasubramanian.
\newblock Gaps in information access in social networks.
\newblock In {\em The World Wide Web Conference, {WWW} 2019, San Francisco, CA,
  USA, May 13-17, 2019}, pages 480--490, 2019.

\bibitem{FishS21}
Benjamin Fish and Luke Stark.
\newblock Reflexive design for fairness and other human values in formal
  models.
\newblock In Marion Fourcade, Benjamin Kuipers, Seth Lazar, and Deirdre~K.
  Mulligan, editors, {\em {AIES} '21: {AAAI/ACM} Conference on AI, Ethics, and
  Society, Virtual Event, USA, May 19-21, 2021}, pages 89--99. {ACM}, 2021.

\bibitem{Friedman:1996tm}
Batya Friedman and Helen Nissenbaum.
\newblock {Bias in Computer Systems}.
\newblock {\em ACM Transactions on Information Systems}, 14(3):330--347,
  September 1996.

\bibitem{fudenberg1993self}
Drew Fudenberg and David~K Levine.
\newblock Self-confirming equilibrium.
\newblock {\em Econometrica: Journal of the Econometric Society}, pages
  523--545, 1993.

\bibitem{gandy2019toward}
Oscar~H Gandy~Jr and Selena Nemorin.
\newblock Toward a political economy of nudge: smart city variations.
\newblock {\em Information, Communication \& Society}, 22(14):2112--2126, 2019.

\bibitem{GillenJKR18}
Stephen Gillen, Christopher Jung, Michael~J. Kearns, and Aaron Roth.
\newblock Online learning with an unknown fairness metric.
\newblock In Samy Bengio, Hanna~M. Wallach, Hugo Larochelle, Kristen Grauman,
  Nicol{\`{o}} Cesa{-}Bianchi, and Roman Garnett, editors, {\em Advances in
  Neural Information Processing Systems 31: Annual Conference on Neural
  Information Processing Systems 2018, NeurIPS 2018, December 3-8, 2018,
  Montr{\'{e}}al, Canada}, pages 2605--2614, 2018.

\bibitem{Hardt:2016wv}
Moritz Hardt, Eric Price, and Nathan Srebro.
\newblock {Equality of Opportunity in Supervised Learning}.
\newblock In {\em 30th Conference on Neural Information Processing Systems},
  pages 1--9, Barcelona, 2016.

\bibitem{HeidariFGK18}
Hoda Heidari, Claudio Ferrari, Krishna~P. Gummadi, and Andreas Krause.
\newblock Fairness behind a veil of ignorance: {A} welfare analysis for
  automated decision making.
\newblock In Samy Bengio, Hanna~M. Wallach, Hugo Larochelle, Kristen Grauman,
  Nicol{\`{o}} Cesa{-}Bianchi, and Roman Garnett, editors, {\em Advances in
  Neural Information Processing Systems 31: Annual Conference on Neural
  Information Processing Systems 2018, NeurIPS 2018, December 3-8, 2018,
  Montr{\'{e}}al, Canada}, pages 1273--1283, 2018.

\bibitem{Heidari:2018us}
Hoda Heidari, Michele Loi, Krishna~P. Gummadi, and Andreas Krause.
\newblock {A Moral Framework for Understanding of Fair ML through Economic
  Models of Equality of Opportunity}.
\newblock pages 1--13, September 2018.

\bibitem{Anonymous:4jyXmvug}
Anna~Lauren Hoffmann.
\newblock {Data Violence and How Bad Engineering Choices Can Damage Society},
  April 2018.

\bibitem{Hoffmann:2019hr}
Anna~Lauren Hoffmann.
\newblock {Where fairness fails: data, algorithms, and the limits of
  antidiscrimination discourse}.
\newblock {\em Information, Communication {\&} Society}, 22(7):900--915, May
  2019.

\bibitem{hu2020direct}
Lily Hu.
\newblock Direct effects.
\newblock {\em Phenomenal World}, 2020.

\bibitem{hu2020s}
Lily Hu and Issa Kohler-Hausmann.
\newblock What's sex got to do with fair machine learning?
\newblock {\em arXiv preprint arXiv:2006.01770}, 2020.

\bibitem{Hutchinson:2019ce}
Ben Hutchinson and Margaret Mitchell.
\newblock {50 Years of Test (Un)fairness}.
\newblock In {\em Conference on Fairness, Accountability, and Transparency
  2019}, pages 49--58, New York, New York, USA, 2019. ACM Press.

\bibitem{Ilvento20}
Christina Ilvento.
\newblock Metric learning for individual fairness.
\newblock In Aaron Roth, editor, {\em 1st Symposium on Foundations of
  Responsible Computing, {FORC} 2020, June 1-3, 2020, Harvard University,
  Cambridge, MA, {USA} (virtual conference)}, volume 156 of {\em LIPIcs}, pages
  2:1--2:11. Schloss Dagstuhl - Leibniz-Zentrum f{\"{u}}r Informatik, 2020.

\bibitem{JafariNaimi:2015en}
Nassim JafariNaimi, Lisa Nathan, and Ian Hargraves.
\newblock {Values as Hypotheses: Design, Inquiry, and the Service of Values}.
\newblock {\em Design Issues}, 31(4):91--104, October 2015.

\bibitem{JungKNRSS19}
Christopher Jung, Michael~J. Kearns, Seth Neel, Aaron Roth, Logan Stapleton,
  and Zhiwei~Steven Wu.
\newblock Eliciting and enforcing subjective individual fairness.
\newblock {\em CoRR}, abs/1905.10660, 2019.

\bibitem{kamiran2009classifying}
Faisal Kamiran and Toon Calders.
\newblock Classifying without discriminating.
\newblock In {\em 2009 2nd International Conference on Computer, Control and
  Communication}, pages 1--6. IEEE, 2009.

\bibitem{KasyA21}
Maximilian Kasy and Rediet Abebe.
\newblock Fairness, equality, and power in algorithmic decision-making.
\newblock In Madeleine~Clare Elish, William Isaac, and Richard~S. Zemel,
  editors, {\em FAccT '21: 2021 {ACM} Conference on Fairness, Accountability,
  and Transparency, Virtual Event / Toronto, Canada, March 3-10, 2021}, pages
  576--586. {ACM}, 2021.

\bibitem{KleinbergMR17}
Jon~M. Kleinberg, Sendhil Mullainathan, and Manish Raghavan.
\newblock Inherent trade-offs in the fair determination of risk scores.
\newblock In Christos~H. Papadimitriou, editor, {\em 8th Innovations in
  Theoretical Computer Science Conference, {ITCS} 2017, January 9-11, 2017,
  Berkeley, CA, {USA}}, volume~67 of {\em LIPIcs}, pages 43:1--43:23. Schloss
  Dagstuhl - Leibniz-Zentrum f{\"{u}}r Informatik, 2017.

\bibitem{kohler2018eddie}
Issa Kohler-Hausmann.
\newblock Eddie murphy and the dangers of counterfactual causal thinking about
  detecting racial discrimination.
\newblock {\em Nw. UL Rev.}, 113:1163, 2018.

\bibitem{KusnerLRS17}
Matt~J. Kusner, Joshua~R. Loftus, Chris Russell, and Ricardo Silva.
\newblock Counterfactual fairness.
\newblock In Isabelle Guyon, Ulrike von Luxburg, Samy Bengio, Hanna~M. Wallach,
  Rob Fergus, S.~V.~N. Vishwanathan, and Roman Garnett, editors, {\em Advances
  in Neural Information Processing Systems 30: Annual Conference on Neural
  Information Processing Systems 2017, December 4-9, 2017, Long Beach, CA,
  {USA}}, pages 4066--4076, 2017.

\bibitem{Lee:2018ik}
Min~Kyung Lee.
\newblock {Understanding perception of algorithmic decisions: Fairness, trust,
  and emotion in response to algorithmic management}.
\newblock {\em Big Data {\&} Society}, 5(1):205395171875668--16, March 2018.

\bibitem{LiuDRSH19}
Lydia~T. Liu, Sarah Dean, Esther Rolf, Max Simchowitz, and Moritz Hardt.
\newblock Delayed impact of fair machine learning.
\newblock In Sarit Kraus, editor, {\em Proceedings of the Twenty-Eighth
  International Joint Conference on Artificial Intelligence, {IJCAI} 2019,
  Macao, China, August 10-16, 2019}, pages 6196--6200. ijcai.org, 2019.

\bibitem{lowther_2016}
George Lowther.
\newblock The section theorems, Nov 2016.
\newblock \url{https://almostsuremath.com/2016/11/29/the-section-theorems/}.

\bibitem{MouzannarOS19}
Hussein Mouzannar, Mesrob~I. Ohannessian, and Nathan Srebro.
\newblock From fair decision making to social equality.
\newblock In danah boyd and Jamie~H. Morgenstern, editors, {\em Proceedings of
  the Conference on Fairness, Accountability, and Transparency, FAT* 2019,
  Atlanta, GA, USA, January 29-31, 2019}, pages 359--368. {ACM}, 2019.

\bibitem{Park:2019fk}
Sora Park and Justine Humphry.
\newblock {Exclusion by design: intersections of social, digital and data
  exclusion}.
\newblock {\em Information, Communication {\&} Society}, 22(7):934--953, May
  2019.

\bibitem{Passi:2019ey}
Samir Passi and Solon Barocas.
\newblock {Problem Formulation and Fairness}.
\newblock In {\em Conference on Fairness, Accountability, and Transparency
  2019}, pages 39--48, New York, New York, USA, 2019. ACM Press.

\bibitem{pearl2009causal}
Judea Pearl.
\newblock Causal inference in statistics: An overview.
\newblock {\em Statistics surveys}, 3:96--146, 2009.

\bibitem{Phelps:1972tu}
Edmund~S Phelps.
\newblock {The Statistical Theory of Racism and Sexism}.
\newblock {\em The American Economic Review}, 62(4):659--661, September 1972.

\bibitem{RaghavanBKL20}
Manish Raghavan, Solon Barocas, Jon~M. Kleinberg, and Karen Levy.
\newblock Mitigating bias in algorithmic hiring: evaluating claims and
  practices.
\newblock In Mireille Hildebrandt, Carlos Castillo, L.~Elisa Celis, Salvatore
  Ruggieri, Linnet Taylor, and Gabriela Zanfir{-}Fortuna, editors, {\em FAT*
  '20: Conference on Fairness, Accountability, and Transparency, Barcelona,
  Spain, January 27-30, 2020}, pages 469--481. {ACM}, 2020.

\bibitem{Rawls:2009ux}
John Rawls.
\newblock {\em {A Theory of Justice}}.
\newblock Harvard University Press, Cambridge, MA, 2009.

\bibitem{Roemer:2002gx}
John~E Roemer.
\newblock {Equality of opportunity: A progress report}.
\newblock {\em Social Choice and Welfare}, 19(2):455--471, February 2002.

\bibitem{Savas_2010}
Gokhan Savas.
\newblock Social inequality at low-wage work in neo-liberal economy: The case
  of women of color domestic workers in the united states.
\newblock {\em Race, Gender \& Class}, 17(3/4):314?326, 2010.

\bibitem{Scheffler:2013vx}
Samuel Scheffler.
\newblock {\em {Equality and Tradition: Questions of Value in Moral and
  Political Theory}}.
\newblock Oxford University Press, Oxford and New York, 2013.

\bibitem{Schemmel:2011ic}
Christian Schemmel.
\newblock {Distributive and relational equality}.
\newblock {\em Politics, Philosophy {\&} Economics}, 11(2):123--148, September
  2011.

\bibitem{Scheuerman_Denton_Hanna_2021}
Morgan~Klaus Scheuerman, Emily Denton, and Alex Hanna.
\newblock Do datasets have politics? disciplinary values in computer vision
  dataset development.
\newblock In {\em Proceedings of the 2021 Conference on Computer Supported
  Cooperative Work (CSCW ?21)}, 2021.

\bibitem{SelbstBFVV19}
Andrew~D. Selbst, danah boyd, Sorelle~A. Friedler, Suresh Venkatasubramanian,
  and Janet Vertesi.
\newblock Fairness and abstraction in sociotechnical systems.
\newblock In danah boyd and Jamie~H. Morgenstern, editors, {\em Proceedings of
  the Conference on Fairness, Accountability, and Transparency, FAT* 2019,
  Atlanta, GA, USA, January 29-31, 2019}, pages 59--68. {ACM}, 2019.

\bibitem{Sen:1992ue}
Amartya Sen.
\newblock {\em {Inequality Reexamined}}.
\newblock Harvard University Press, Cambridge, MA, 1992.

\bibitem{Verma:2018hw}
Sahil Verma and Julia Rubin.
\newblock {Fairness definitions explained}.
\newblock In {\em 2018 ACM/IEEE International Workshop on Software Fairness},
  pages 1--7, New York, New York, USA, 2018. ACM Press.

\bibitem{viljoen2020democratic}
Salom{\'e} Viljoen.
\newblock Democratic data: A relational theory for data governance.
\newblock {\em Available at SSRN 3727562}, 2020.

\bibitem{Walzer:1983wu}
Michael Walzer.
\newblock {\em {Spheres of Justice: A Defense of Pluralism and Equality}}.
\newblock Basic Books, New York, NY, 1983.

\bibitem{Woodruff:2018if}
Allison Woodruff, Sarah~E Fox, Steven Rousso-Schindler, and Jeffrey Warshaw.
\newblock {A Qualitative Exploration of Perceptions of Algorithmic Fairness}.
\newblock In {\em Extended Abstracts of the 2018 CHI Conference}, pages 1--14,
  New York, New York, USA, 2018. ACM Press.

\bibitem{Zuboff:2019uf}
Shoshana Zuboff.
\newblock {\em {The Age of Surveillance Capitalism}}.
\newblock The Fight for a Human Future at the New Frontier of Power.
  PublicAffairs/Hachette, New York, 2019.

\end{thebibliography}

\appendix

\end{document}